\documentclass[english]{ccdconf}

\usepackage{amsmath}
\usepackage{graphics} 
\usepackage{epsfig}   
\usepackage{mathptmx} 
\usepackage{times}    
\usepackage{amsmath}  
\usepackage{amssymb, mathtools}  
\usepackage{psfrag}   
\usepackage{url}
\usepackage[english]{babel}
\usepackage{amsfonts}
\usepackage[noend]{algpseudocode}
\usepackage{bbm}
\usepackage[ruled, vlined, linesnumbered]{algorithm2e}
\usepackage{cite}
\usepackage{graphicx}
\usepackage{lipsum}
\usepackage{booktabs} 

\usepackage{tikz}
\usetikzlibrary{shapes.geometric, arrows}
\tikzstyle{start} = [rectangle, rounded corners, minimum width=0.85\columnwidth, minimum height=1cm, text centered, draw=black, thick, fill=orange!60]
\tikzstyle{decision}  = [diamond, shape aspect=2.2, minimum width=2.6cm, minimum height=0.6cm, inner xsep=0,text centered, thick, draw=black, fill=green!40]
\tikzstyle{proc} = [rectangle, rounded corners, minimum width=0.35\columnwidth, minimum height=1cm, text centered, draw=black, thick, fill=orange!60]
\tikzstyle{proc2} = [rectangle, rounded corners, thick, minimum width=0.4\columnwidth, minimum height=1cm, text centered, draw=black, fill=orange!60]
\tikzstyle{end} = [rectangle, rounded corners, minimum width=0.35\columnwidth, minimum height=1cm, text centered, draw=black, thick, fill=blue!40]
\tikzstyle{arrow}     = [thick, ->, >=stealth]



\usepackage{amsthm}
\newtheorem{assumption}{Assumption}
\newtheorem{proposition}{Proposition}
\newtheorem{remark}{Remark}
\newtheorem{theorem}{Theorem}
\newtheorem{lemma}{Lemma}
\newcommand{\R}{\ensuremath{\mathbb{R}}}
\DeclarePairedDelimiterX{\norm}[1]{\lVert}{\rVert}{#1}

\newcommand{\CASE}[1]{\STATE \textbf{case} #1\textbf{:} \begin{ALC@g}}
	\newcommand{\ENDCASE}{\end{ALC@g}}

\newcommand{\DEFAULT}{\STATE \textbf{default:} \begin{ALC@g}}
	\newcommand{\ENDDEFAULT}{\end{ALC@g}}
\newcommand{\DEFAULTLINEcompareCLFCBF}[1]{\STATE \textbf{default:}}

\begin{document}
	\title{On Estimating the Probabilistic Region of Attraction for Partially Unknown Nonlinear Systems: An Sum-of-Squares Approach}
	\author{Hejun Huang\aref{cuhk1}, Dongkun Han\aref{cuhk2}}
	
	\affiliation[cuhk1]{Department of Mechanical and Automation Engineering, The Chinese University of Hong Kong, HKSAR, China
		\email{hjhuang@mae.cuhk.edu.hk}}
	\affiliation[cuhk2]{Department of Mechanical and Automation Engineering, The Chinese University of Hong Kong, HKSAR, China
		\email{dkhan@mae.cuhk.edu.hk}}
	
	\maketitle
	
	\begin{abstract}
		Estimating the region of attraction for partially unknown nonlinear systems is a challenging issue. In this paper, we propose a tractable method to generate an estimated region of attraction with probability bounds, by searching an optimal polynomial barrier function. Techniques of Chebyshev interpolants, Gaussian processes and sum-of-squares programmings are used in this paper. To approximate the unknown non-polynomial dynamics, a polynomial mean function of Gaussian processes model is computed to represent the exact dynamics based on the Chebyshev interpolants. Furthermore, probabilistic conditions are given such that all the estimates are located in certain probability bounds. Numerical examples are provided to demonstrate the effectiveness of the proposed method.
	\end{abstract}
	
	
	
	\section{Introduction}
	
	Tracking the performance of uncertain nonlinear systems is an essential problem of significant research interests. In engineering practices, the concepts of safety and stability are central to these uncertain nonlinear systems in most scenarios, e.g., flight dynamics, bipedal robotics and power systems \cite{glassman2012region, wang2018safe, hsu2015control, 6171041}. Consequently, scientists across multiple disciplines have summarized this type of problem into the analysis of region of attraction (ROA, also called domain of attraction). Estimating the ROA can directly obtain the safety or the stability margin in many practical implementations \cite{chesi2011domain}.
	
	Fruitful results have been obtained to estimate the ROA of fixed nonlinear systems. Exploiting the sublevel set of Lyapunov functions is one of the useful methods \cite{zubov1964methods}. Barrier functions are another powerful tool to guarantee safety such that states avoid entering a specified unsafe region \cite{blanchini1999set}, \cite{ames2019control}. To meet the safety and the stability conditions simultaneously, the quadratic programs have been investigated to find qualified Lyapunov barrier functions \cite{ames2016control}. Meanwhile, sum-of-squares programmings (SOSPs) are proposed to find more permissive results in polynomial systems \cite{parrilo2000structured,manchester2011regions,el2017estimating,wang2018permissive}. However, in many engineering practices, there exist model inaccuracies and unknown disturbances that might influence the system dynamics or even result in operation failures in the worst case.
	
	To handle with the above issue, effective methods have been proposed regarding partially unknown systems. First, Lyapunov-based methods are developed and extended to uncertain systems, where a Lyapunov certified ROA (LCROA) estimation is conducted for uncertain, polynomial systems \cite{han2016cdc}, \cite{iannelli2018equilibrium}. Furthermore, learning-based methods have been studied in the literature. Among these, the use of Gaussian processes (GP) is shown to be a promising approach to quantify the uncertainty in the stochastic process \cite{vinogradska2017stability,umlauft2018uncertainty}. The idea of unifying SOSP and GP naturally arises to compute the LCROA \cite{berkenkamp2016safe,umlauft2017learning,devonport2020bayesian}. 
	
	Inspired by the work in \cite{jagtap2020control} that firstly uses GP to compute the barrier certified ROA (BCROA), based on our previous work in \cite{wang2018permissive, han2016cdc2}, this paper combines SOSPs and GP to estimate the optimal BCROA of partially unknown nonlinear systems. Different from the learned polynomial systems in \cite{devonport2020bayesian}, we use Chebyshev interpolants and polynomial mean functions of GP models to find the optimal BCROA rather than LCROA with relaxed assumptions. The polynomial mean function is proven to be more flexible to match other nonlinear kernels instead of polynomial kernel. To the best of our knowledge, this paper is the first to compute barrier functions via SOSPs in partially unknown nonlinear systems.
	
	The main contributions of this paper are threefold. First, a learned polynomial system is built with probability bounds. Second, a positive but safe sample policy is developed to prepare appropriate prior information for higher accuracy of the GP model. Third, a tractable method based on SOSP is proposed to compute the probabilistic optimal BCROA.	
	
	\section{Preliminary} \label{sec:prelimary}
	
	Consider an autonomous system as follows,
	\begin{equation}
		\label{eqn:sysstate}
		\dot{x} = \underbrace{f(x)+g(x)}_{Known}+\underbrace{d(x)}_{\textit{Unknown}},
	\end{equation}	
	\noindent where $x \in \mathcal{X} \subseteq \mathbb{R}^{n}$ denotes the state, $f, g:\mathbb{R}^{n}\rightarrow \mathbb{R}^{n}$ denote the polynomial and the non-polynomial term respectively, and $d:\mathbb{R}^{n}\rightarrow \mathbb{R}^{n}$ denotes the unknown term. All the terms in (\ref{eqn:sysstate}) are Lipschitz continuous.
	
	Without losing generality, systems with a single equilibrium are considered, and this equilibrium could be transformed to the origin via variables substitution \cite{Khalil:1173048}. 
	
	\begin{assumption}
		\label{asp:1}
		The origin ($x=0$) is a single stable equilibrium of (\ref{eqn:sysstate}), that is $f(0)=g(0)=d(0)=0$. $\hfill\square$
	\end{assumption}

	To model the system dynamics in a Bayesian framework as developed in \cite{3569}, we define a prior distribution of noise $[\epsilon_1,\epsilon_2,\dots, \epsilon_k]^{\mathrm{T}}$ over $k$ measurements $[\dot{x_1}, \dot{x_2}, \dots, \dot{x_k}]^{\mathrm{T}}$.

	\begin{assumption}
		The noise $[\epsilon_1,\epsilon_2,\dots, \epsilon_k]^{\mathrm{T}}$ over the system measurements $[\dot{x}_1,\dot{x}_2,\dots, \dot{x}_k]^{\mathrm{T}}$ of (\ref{eqn:sysstate}) is uniformly bound by $\sigma_n$, i.e., $\{\epsilon_i \sim\mathcal{N}(0,\sigma_{n}^2)\}_{i=1}^{k}$.  $\hfill\square$
	\end{assumption}

	The known dynamics $f(x)+g(x)$ can be computed directly. Thus, $d(x)$ can be obtained by subtracting $f(x)+g(x)$ from $\dot{x}$. In this work, we restrict our attention to the system (\ref{eqn:sysstate}) with bounded $d(x)$:

	\begin{assumption}
		The unknown term $d(x)$ in (\ref{eqn:sysstate}) exists a bounded norm in the reproducing kernel Hilbert space (RKHS), i.e., $\Vert d(x)\Vert\leq c_g$, where $c_g$ is a constant. $\hfill\square$
	\end{assumption}

	 As \cite{jagtap2020control} introduced, the RKHS is a Hilbert space of square integrable functions that contains functions of the form $l(x)=\sum_i \alpha_i k(x, x_i)$, where $\alpha_i\in\R$ denotes coefficient and $k:\mathcal{X}\times \mathcal{X}\rightarrow\mathbb{R}^+_0$ denotes a symmetric positive definite kernel function of states $x, x_i$. For more details about the RKHS norm, we kindly refer interested readers to \cite{paulsen2016introduction}.
	
	\subsection{Barrier Functions}\label{sec:Barrier function}
	
	To introduce the barrier function, let us first consider a simple autonomous system as follows,
	\begin{equation}
		\label{eqn:polynomialsystem}
		\begin{aligned}
			\dot{x}=f(x),
		\end{aligned}
	\end{equation}
	\noindent where $f(x)$ is locally Lipschitz continuous with a single stable equilibrium point at the origin. As we mentioned, the safety and stability of (\ref{eqn:polynomialsystem}) can be guaranteed by ROA. 
	
	State trajectories starting inside the barrier function certified ROA (BCROA) $\mathcal{L} =\{x\in\mathcal{X} \vert  h(x)\geq 0\}$ will never enter into the unsafe region $\mathcal{L}_U=\mathcal{X}\backslash\mathcal{L}$ as defined,
	
	\begin{equation} \label{Def:BarrierCertificate}
		\begin{aligned}
			& \forall x\in \mathcal{L} \;\;\;\; && h(x)\geq 0, \\
			& \forall x\in \mathcal{L} && \frac{\partial h}{\partial x}f(x)\geq 0,\\
			& \forall x\in \mathcal{L}_U && h(x) < 0. \\
		\end{aligned}
	\end{equation}
	
	\noindent Meanwhile, the Lyapunov function $V(x)$ certified ROA  (LCROA) $\mathcal{L}_V=\{x\in\mathcal{X}\, \vert V(x)\leq c\}$ denotes a sublevel set of $V(x)$. The state trajectories inside $\mathcal{L}_V$ will always converge to the origin. The following lemma demonstrates a relationship between $\mathcal{L}$ and $\mathcal{L}_V$.
	
	\begin{lemma}(Lemma 3.2 of \cite{wang2018permissive})
		\label{lem:permissive_bc}
		Given an autonomous dynamical system (\ref{eqn:polynomialsystem}) that is asymptotically stable at the origin, the estimate of ROA by barrier function is no smaller than the estimate by Lyapunov function, i.e., $\mathcal{L}_v\subseteq \mathcal{L}$. $\hfill\square$
	\end{lemma}
	
	\noindent The details about how to obtain a Lyapunov maximum sublevel $c^*$ and an optimal barrier function $h^*(x)$ are introduced \cite{wang2018permissive}. In Example 1, we will illustrate the comparison of LCROA and BCROA. 
	
	\textit{Example $1$}:~ Consider an autonomous system
	\begin{eqnarray}
		\label{eqn:example1}
		\begin{bmatrix} \dot{x}_1 \\ \dot{x}_2 \end{bmatrix} = 
		\begin{bmatrix}
			-x_{1}^3-x_{1}x_{2}^2\\	-x_{2}-x_{1}^2x_{2}
		\end{bmatrix}.
	\end{eqnarray}
	
	Let $\mathcal{A}_1$ denote the LCROA and $\mathcal{A}_2$ denote the BCROA. In (\ref{eqn:example1}), $\mathcal{A}_1 = \{x \, \vert \, V_1(x)\leq 1.063\}$ is established by a $4^{th}$ degree polynomial $V_1(x) =  x_{1}^{4} + x_{2}^{4}+x_{1}^{2}x_{2}^{2}+x_{1}^{2} + x_{2}^{2}+x_{1}x_{2}$, and $\mathcal{A}_2= \{x \, \vert \, h_1^*(x)\geq 0\}$ is obtained by another $4^{th}$ degree polynomial barrier function $h_1^*(x)$ based on $\mathcal{A}_1$.
	
	\begin{figure}[ht] 
		\centering
		\includegraphics[width=0.85\linewidth]{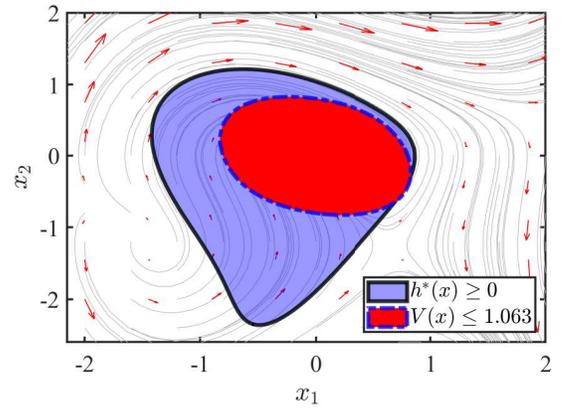}
		\caption{Comparison of $\mathcal{A}_1$ and $\mathcal{A}_2$. The red region enclosed by a dashed blue ellipse depicts $\mathcal{A}_1$, while the light blue region enclosed by the solid black line depicts $\mathcal{A}_2$. }
		\label{fig:fig1}
	\end{figure}
	
	In Figure \ref{fig:fig1}, $\mathcal{A}_{2}$ is significantly larger than $\mathcal{A}_{1}$, which means the former allows more states to be considered in (\ref{eqn:example1}). In this paper, we will leverage the superiority of the barrier function to obtain a permissively safe and stable guarantee. 
	
	\subsection{Gaussian Processes}\label{sec:Gaussian Processes}
	Gaussian processes (GP) provide a non-parametric regression method to capture the unknown dynamics in Bayesian framework. In the GP model, every variable is associated with a normal distribution, where the prior and the posterior GP model of these variables obey a joint Gaussian distribution \cite{3569}. A GP model is characterized by the mean function $m(x)$ and the covariance function $k(x, x^{\prime})$, where $k(x,x^{\prime})$ computes the similarity between any two states $x, \,x^{\prime}\in\mathcal{X}$. In most cases, the covariance function $k(x,x^{\prime})$ is also called a kernel. The measurements of $d(x)$ in (\ref{eqn:sysstate}) can be used to construct a GP model with polynomial mean functions as follows.
	
	\begin{proposition}
		\label{prop:polynomial mean function}
		Suppose there exist $k$ measurements of $d(x)$ in (\ref{eqn:sysstate}) that satisfy Assumption 2, 3. Then, the following GP model of $d(x)$ can be established with polynomial mean function $m(x_*)$ and covariance function $k(x, x_*)$ as,
		\begin{equation}
			\label{eqn:optimize_output_distribution_mean}
			\begin{aligned}
				m(x_*) &= \varphi(x_*)^{\mathrm{T}} w, \\
				k(x, x_*) &= k(x_*, x_*)-k_*^{\mathrm{T}}(K+\sigma_n^2I)^{-1}k_*.
			\end{aligned}
		\end{equation}	
		\noindent where $x_*$ is a query state, $\varphi(x_*)$ is a monomial vector, $w$ is a coefficient vector, $[K]_{(i,j)}=k(x_i, x_j)$ is a kernel Gramian matrix and $k_*=[k(x_1,x_*), k(x_2,x_*), \dots, k(x_w, x_*)]^{\mathrm{T}}$.
	\end{proposition}
	\begin{proof}
		See Appendix.
	\end{proof}
	
	\begin{remark} 
		Proposition 1 supports us to define the degree of polynomial $m(x_*)$ in (\ref{eqn:optimize_output_distribution_mean}). Besides, it establishes a feasible link between polynomial mean functions and other flexible kernels, including polynomial kernel \cite{devonport2020bayesian}.  $\hfill\square$
	\end{remark}
	
	\section{Learning the Partially Unknown System}\label{sec:EstimateDynamic}
	To illustrate our work about learning the partially unknown dynamics (\ref{eqn:sysstate}) in polynomial form, we propose the approximated probabilistic model in (\ref{eqn:gp_learned_system_2}), which combines Chebyshev interpolants and a GP model as shown in the first two subsections. We also introduce a covariance oriented safe sample policy based on GP to learn the dynamics consistently in the third subsection.
	

	\subsection{Chebyshev Interpolants}\label{sec:Chebyshev}
	Chebyshev interpolants provide a useful way to approximate a class of nonlinear functions with a bounded remainder \cite{trefethen2019approximation}. The term $g(x)$ in (\ref{eqn:sysstate}) can be approximated by the Chebyshev interpolants $P_k$ of degree $k$ in $[-1,1]$ as, 
	\begin{equation}
		\begin{aligned}
			g(x)\approx P_k(x)&=\sum_{i=0}^{k}\alpha_i \tau_i(x)=\left\langle A ,T\right\rangle, \quad i\in [0,k], \\
		\end{aligned}
		\label{eqn:Cheby}  
	\end{equation}
	\noindent where $i$ denotes the subscript of the coefficients vector $A =[\alpha_{0}, \alpha_{1}, \dots, \alpha_{k}]$ and the Chebyshev polynomials vector $T =[\tau_{0}, \tau_{1}, \dots, \tau_{k}]$ that satisfies,
	\begin{equation} \label{eqn:ChebyAT}
		\begin{aligned}
			\alpha_i(x) &= \left\{
			\begin{aligned}
				&\frac{1}{\pi} \int_{-1}^{1} \frac{{g(x)\tau_i(x)}}{{\sqrt{1-x^2}}} dx, \quad i=0,\\
				&\frac{2}{\pi} \int_{-1}^{1} \frac{{g(x)\tau_i(x)}}{{\sqrt{1-x^2}}} dx, \quad i\in[1,k], \\
			\end{aligned}
			\right.\\
			\tau_i(x)&=\cos(i\operatorname{arccos}(x)), \qquad\quad  i\in[0,k].\\
		\end{aligned}
	\end{equation}
	
	\noindent Note that, Chebyshev interpolants are applicable to any arbitrary interval $I=[a, b]$ by the following transformation,
	\begin{equation}
		I(x)=\frac{2x-(b+a)}{b-a}.
		\label{eqn:IntervalT}
	\end{equation} 
		
	
	We define the remainder $\xi= g(x)-P_k(x)$ based on (\ref{eqn:Cheby}). The following inequality from \cite{trefethen2019approximation} declares that the remainder $\xi$ of Chebyshev interpolants is bounded over the domain.

	\begin{lemma}(Theorem 8.2 of \cite{trefethen2019approximation})\label{lm:Approximation}
	Let an analytic function $g(x)$ in $[-1,1]$ be analytically containable to the open Bernstein ellipse $E$, where it satisfies $\vert g(x)\vert \leq m$ for some $m$. Then, its Chebyshev interpolants $P_k(x)$ satisfies
	{\vspace{0pt}\hspace{40pt}
		$
		\begin{aligned}
			\Vert g(x)-P_k(x) \Vert \leq \frac{4m  \rho^{-k}}{\rho -1},\quad k \geq 0.
		\end{aligned}$}
	\hfill$\square$
	\end{lemma}
	\noindent The Bernstein ellipse $E$ has foci $\pm 1$ and major radius $1+\rho$ for all $\rho \ge 0$. For more details about the Bernstein ellipses, we kindly refer to the book by Trefethe \cite{trefethen2019approximation}.
	
	Based on Lemma \ref{lm:Approximation}, the system (\ref{eqn:sysstate}) can be expressed as,
	\begin{equation}
		\begin{aligned}
			\label{eqn:Appro_sys}
			\dot{x}\;\;\;&=\;f(x) + P_k(x) + d_{\xi}(x), \\
		\end{aligned}
	\end{equation}
	
	\noindent where the unknown term $d_{\xi}(x)$ satisfies Assumption 2, 3,
	\begin{equation}
		\begin{aligned}
			\label{eqn:Appro_d}
			d_{\xi}(x)=d(x)+\xi.\\
		\end{aligned}
	\end{equation}
	
	The approximation accuracy of Chebyshev interpolants is linearly dependent on the degree. We will illustrate a comparison in the following example.
	
	\textit{Example $2$}:~ Consider a nonlinear function $y = \sqrt{\vert x-3 \vert}$ approximated by the Chebyshev interpolants of degree $4$ and $80$ in $[0,6]$, the higher degree approximation is more accurate with less bounded remainder $\xi$ in Figure \ref{fig:fig2}. 
	
	\begin{figure}[ht]
		\centering 
		\begin{minipage}{0.23\textwidth}
			\centering 
			\includegraphics[width=\textwidth]{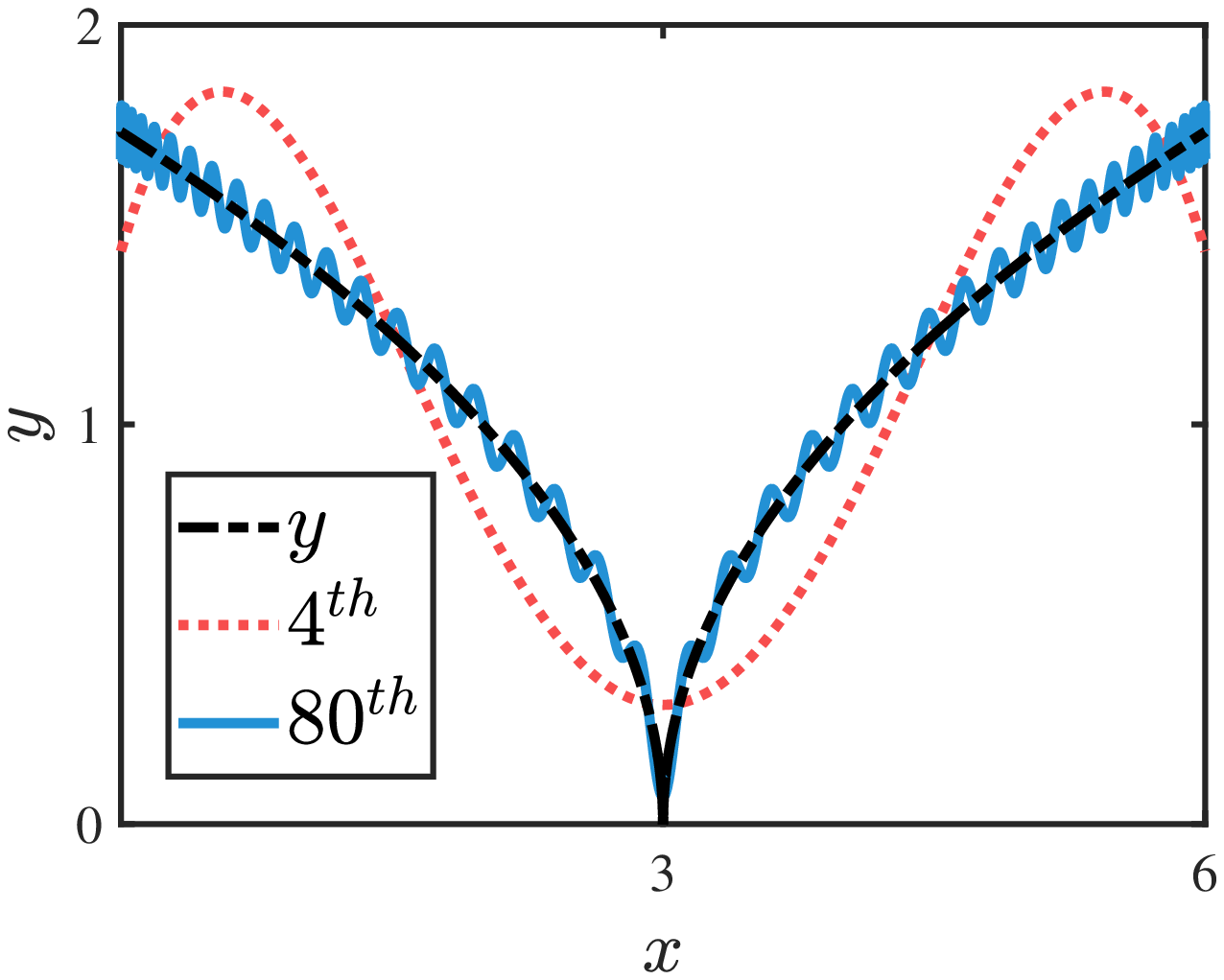}
			\centerline{(a)}
		\end{minipage}
		\begin{minipage}{0.23\textwidth}
			\centering 
			\includegraphics[width=\textwidth]{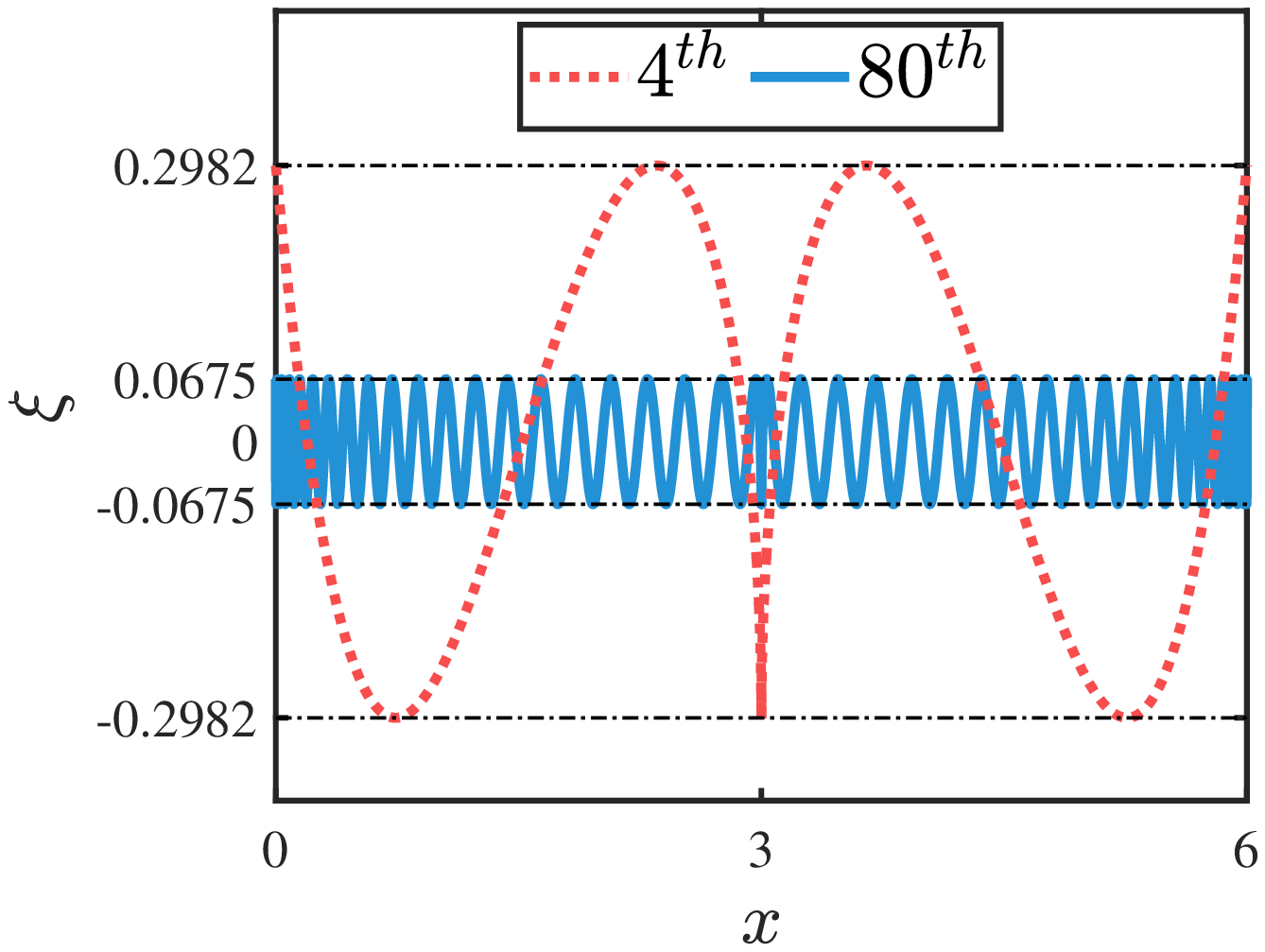}
			\centerline{(b)}
		\end{minipage}
		\caption{The red dot line and the blue solid line depict the results of the Chebyshev interpolants of degree $4$ and $80$, respectively. $(a)$ The black dashed line depicts the real value of $y$. $(b)$ The black dashed lines denote the boundary of different Chebyshev interpolants remainder $\xi$.}
		\label{fig:fig2}
	\end{figure}
	

	\subsection{Probability Bounds on the GP model}
	The following result provides the probability bounds in the distribution of a GP model.
	
	\begin{lemma}(Lemma 2 of \cite{umlauft2018uncertainty}) 
		\label{lem:krause_inequality}
		Suppose there exist $k$ measurements $\{(x_1,y_1), (x_2,y_2), \dots, (x_k, y_k)\}$ of a bounded function $h(x)\in \mathcal{H}$, where $x$ denotes the state, $y$ denotes the measurement of $h(x)$ that corrupted with noise $\{\epsilon_1,\epsilon_2,\cdots \epsilon_k\}\sim\mathcal{N}(0,\sigma_n^2)$, and $\mathcal{H}$ denotes the RKHS. Let $\delta\in(0,1)$. For $h(x)$ and its inferred GP model $(m_h(x), \sigma^2_h(x))$ holds,
		\begin{equation} \label{eq:hhatbound}
			\mathbb{P}\{\Vert h(x)-m_{h}(x) \Vert  \leq \Vert \beta\Vert \Vert  \sigma_{h}(x)\Vert ,\forall x \in \mathcal{X}\}\geq (1-\delta)^k,
		\end{equation}
		\noindent where $[\beta]_k =\sqrt{2\norm{h_k}^2 + 300\gamma_k\log^3(k/\delta)}$ is a discounting factor, $\Vert\cdot\Vert$ is a norm of $\mathcal{H}$, and $\gamma_k$ is a factor denoting the maximum mutual information over the measurements.$\hfill\square$
	\end{lemma}
	
	
	\begin{remark}
		The value of $\gamma_k$ is related to the type of kernels, e.g., RBF kernel, Matérn kernel and linear kernel from \cite{srinivas2009gaussian}. Thus, its sublinear dependent term $\beta_{k}$ can be regarded as a constant \cite{berkenkamp2016safe}. In \cite{umlauft2018uncertainty}, a more correlated statement yields that with more appropriate prior information, the value of $\sigma_{h}(x)$ will decrease such that the dynamics $h(x)$ can be represented by $m_h(x)$ with fewer differences.$\hfill\square$
	\end{remark}
	
	The unknown term $d_\xi(x)$ in (\ref{eqn:Appro_d}) can be learned by using a GP model. Combining previous results of the system (\ref{eqn:Appro_sys}), a probabilistic statement toward the exact dynamics,
	\begin{equation}
		\label{eqn:gp_learned_system_2}
		\begin{aligned}
			\dot{x}= f(x)+P_k(x)+m_{d_\xi}(x),
		\end{aligned}
	\end{equation} 
	
	\noindent holds with probability greater or equal to $(1-\delta)^k, \delta\in(0,1)$. Therefore, we can obtain a high probability statement (\ref{eqn:gp_learned_system_2}) of (\ref{eqn:sysstate}) based on Lemma 3, which will be needed in the analysis of the proposed ROA estimate.

	\subsection{Covariance Oriented Safe Sample Policy}\label{sec:cov}
	
	Computing a high-confidence GP model is closely related to the appropriate prior information. If the mean function $m(x)$ of a GP model is closer to the prior information, $k(x, x_*)$ will decrease to get close to the exact dynamics \cite{3569}. Therefore, handling the high covariance data will accelerate the learning process. Meanwhile, to maintain the prior information is always safe, we propose a positive sample policy inside an existing ROA as follows, 
	
	\begin{equation}
		\label{eqn:samplepolicy}
		\begin{aligned}
			x^* = \mathop{\arg\max}\limits_{x, \psi(t;d_\xi(x))\in \mathcal{L}_k}\quad \sum\limits_{k=0}^{K}{{k}}(x, \psi(t;d_\xi(x))),
		\end{aligned}
	\end{equation}
	
	\noindent where $\sum k(x, \psi(t;d_\xi(x)))$ denotes the covariance of a trajectory $\psi(t;d_\xi(x))$ starting from $x\in\mathcal{L}_{k}$, $\mathcal{L}_{k}$ denotes the estimate ROA and $k$ denotes the index number of the total episode $K$. The trajectory $\psi(t;d_\xi(x^*))$ of the result in (\ref{eqn:samplepolicy}) can be used to enrich the GP prior model recursively. The example below shows a comparison of the dynamics between the original system (\ref{eqn:sysstate}) and a learned system (\ref{eqn:gp_learned_system_2}).
	
	\textit{Example $3$}:~ Consider a partially unknown nonlinear system (\ref{example:ROA}) with unknown term $d(x)$ as
	\begin{eqnarray}\label{example:ROA}
		\begin{aligned}
			\begin{bmatrix} \dot{x}_1 \\ \dot{x}_2 \end{bmatrix} = 
			\begin{bmatrix}
				-x_{1}+x_{2}\\	x_{1}^2x_{2}+1-\sqrt{|\exp{(x_{1})}\cos{(x_{1})}|}
			\end{bmatrix}+
			\begin{bmatrix}
				0\\	d(x)
			\end{bmatrix}.
		\end{aligned}   
	\end{eqnarray}
	
	\begin{figure}[ht]
		\centering 
		\begin{minipage}{0.23\textwidth}
			\centering 
			\includegraphics[width=\textwidth]{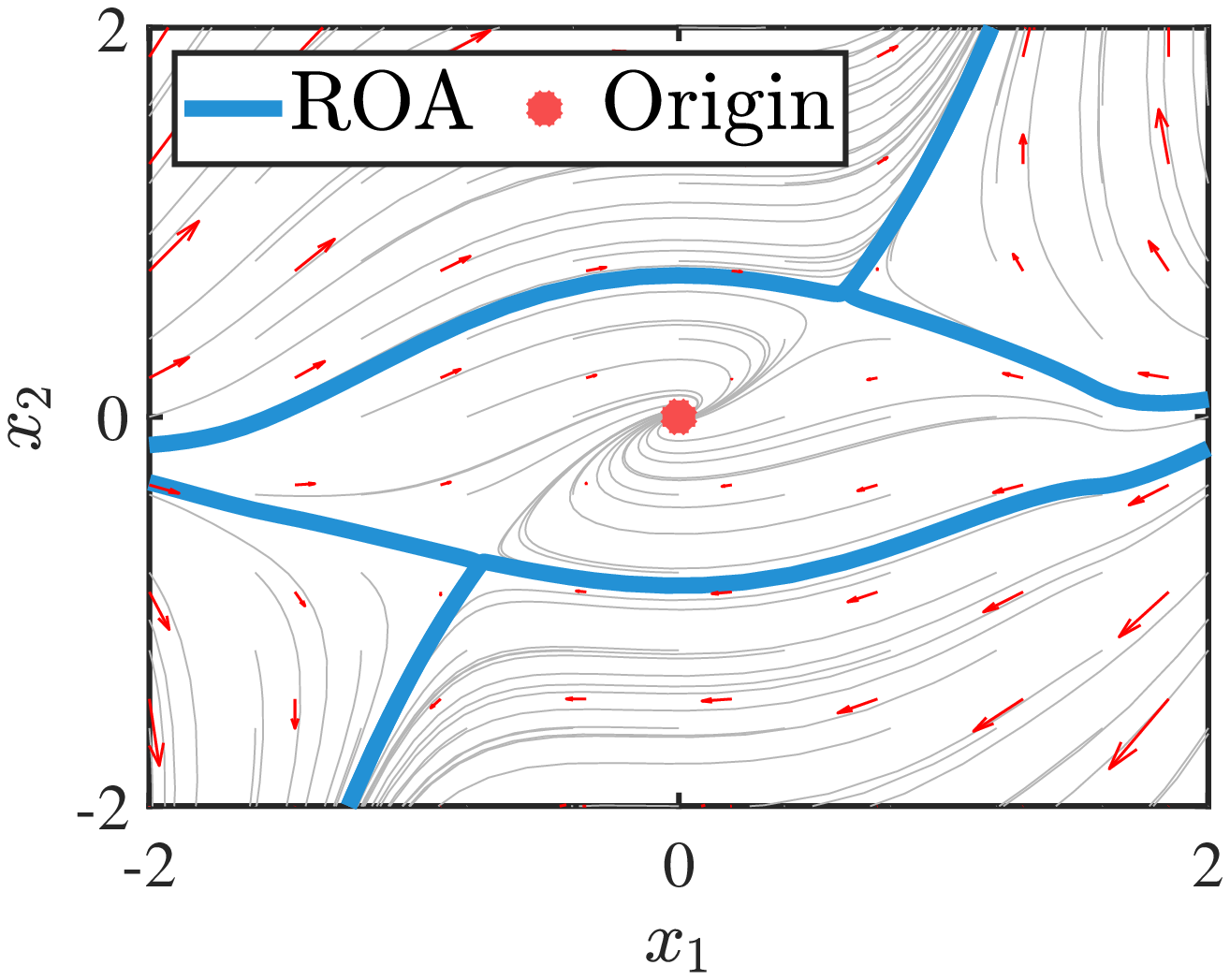}
			\centerline{(a)}
		\end{minipage}
		\begin{minipage}{0.23\textwidth}
			\centering 
			\includegraphics[width=\textwidth]{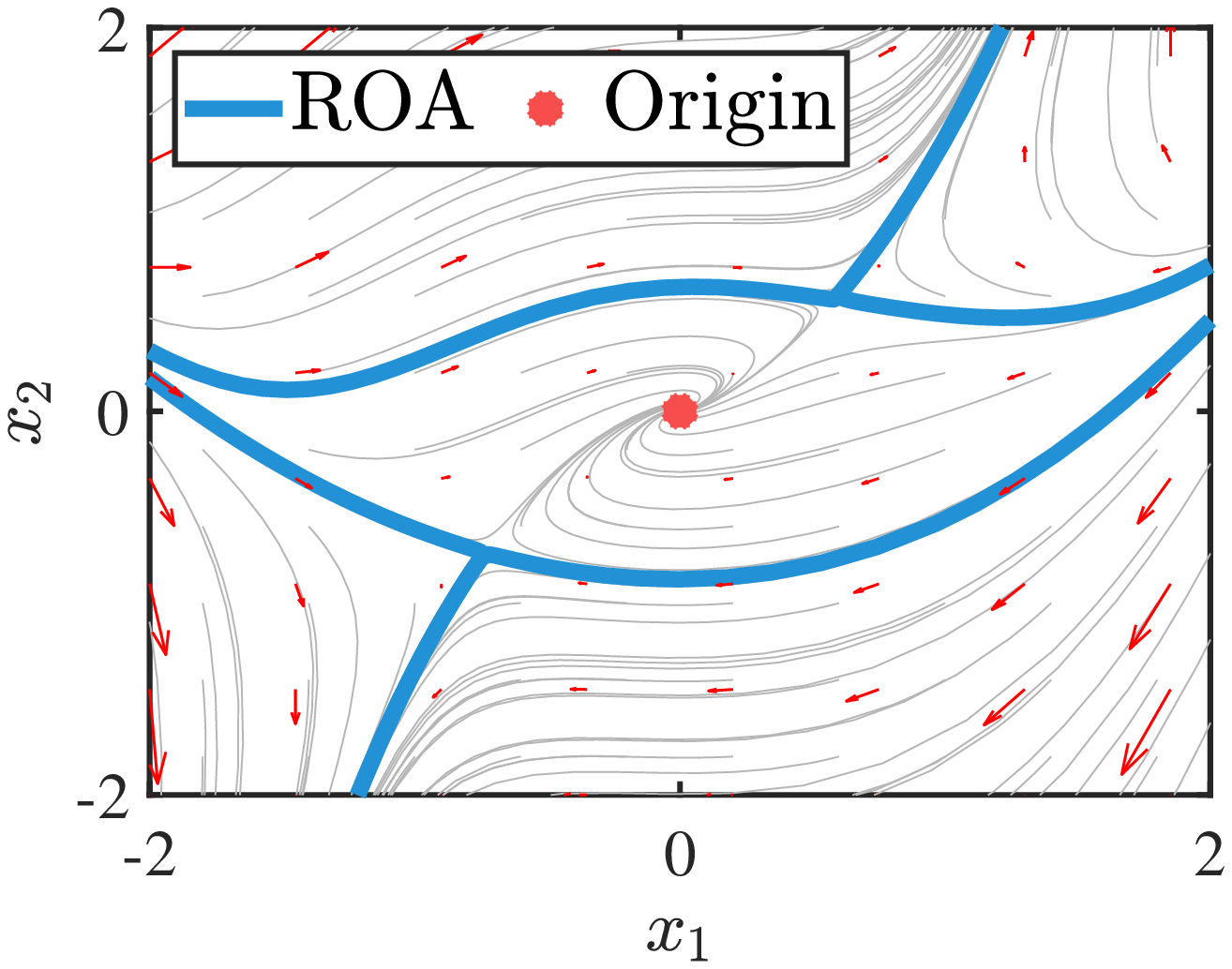}
			\centerline{(b)}
		\end{minipage}
		\caption{$(a)$ The left schematic demonstrates the exact ROA of (\ref{example:ROA}) with $d(x)=0$. $(b)$ The right schematic demonstrates the exact ROA of a learned polynomial system. The blue solid line depicts the ROA boundary, and the red filled circle depicts the equilibrium point.}
		\label{fig:fig_Comparison}
	\end{figure}
	
	\begin{figure}[ht] 
		\centering
		\includegraphics[width=0.85\linewidth]{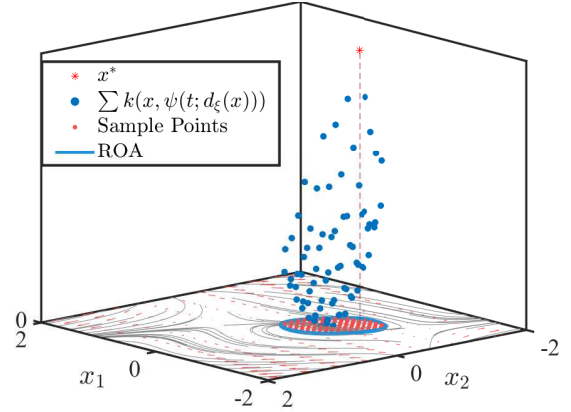}
		\caption{The result of $\sum{{k}}(x, \psi(t;d_\xi(x)))$ in Figure \ref{fig:fig_Comparison}(b). The red star point poses the most uncertain point with maximum value of  $\sum{{k}}$, the blue solid points in the space denote the value of $\sum{{k}}$ starting from the red solid points in $x_1, x_2$ plane and the solid blue eclipse denotes an estimated ROA.  }
		\label{fig:fig2d_2}
	\end{figure}
	
	Figure \ref{fig:fig_Comparison}(b) shows the learned polynomial system dynamics of (\ref{example:ROA}). The non-polynomial term in (\ref{example:ROA}) is approximated by $4^{th}$ degree Chebyshev interpolants in $[-2,2]^2$, while the state trajectory $\psi(t;d_\xi(-0.05,-0.05))$ is used to construct a GP model with $3^{rd}$ degree polynomial mean function. Figure 4 shows the obtained sum of covariance of the learned polynomial system based on ($\ref{eqn:samplepolicy}$). Compared to the points around the equilibrium point, those nearby the ROA edge have larger sum of covariance and can be considered with high uncertainty. Thus, the distribution of the sum of covariance in Figure 4 poses a direction to reduce the uncertainty during the learning process.
	
	\section{Estimating the Region of Attraction}\label{sec:ROAestimate}
	In this section, we will estimate the ROA of the learned polynomial system (\ref{eqn:gp_learned_system_2}) via SOSPs . The relationship of estimates between the learned polynomial system and the true dynamics will also be discussed. 
	
	\subsection{Parameterize the ROA by Barrier Function}
	\label{ap:1}
	
	The parameterization of a barrier function in (\ref{eqn:gp_learned_system_2}) via SOSPs is developed in \cite{wang2018permissive}. To deal with the non-negativity constraints, we will introduce the Positivestellensatz (P-satz) here for further SOSPs implementation.

    Let $\mathcal{P}$ be the set of polynomials and $\mathcal{P}^\text{SOS}$ be the set of sum of squares polynomials, e.g., $P(x)=\sum_{i=1}^{k}p_i^2(x), $ where $p_i(x)\in\mathcal{P}$ and $P(x)\in \mathcal{P}^{SOS}$. 
	
	\begin{lemma}\label{lem:psatz}
		(\cite{putinar1993positive}) For polynomials $\{a_i\}_{i=1}^m$, $\{b_j\}_{j=1}^n$ and $p$, define a set $\mathcal{B}=\{x\in\mathcal{X}:\{a_i(x)\}_{i=1}^m=0,\{b_j(x)\}_{j=1}^n\geq0\}$. Let $\mathcal{B}$ be compact, then $\forall x \in \mathcal{X}, p(x)\geq0$ holds if,
		
		\vspace{7pt}\centering{\hspace{32pt}
			$\left\{
			\begin{array}{l}
				\exists r_1,\dots, r_m \in \mathcal{P}, ~ s_1,\dots, s_n \in \mathcal{P}^{\text{SOS}}, \\
				p-\sum^{m}_{i=1}r_i a_i-\sum^{n}_{j=1}s_j b_j \in \mathcal{P}^{\text{SOS}}.
			\end{array} 
			\right. $}
		\hfill$\square$
	\end{lemma}
	This lemma declared that any strictly positive polynomial $p$ is in the cone that generated by polynomials $\{a_i\}_{i=1}^m$ and $\{b_j\}_{j=1}^n$. Using Lemma \ref{lem:psatz}, the optimal barrier function $h^*(x)$ of the learned polynomial system (\ref{eqn:gp_learned_system_2}) can be searched in the state space via \textbf{3} SOSPs.
	
	\textbf{1.} Obtain the maximum sublevel set $\{x\vert V(x)\leq c^*\}$ of a specified Lyapunov function $V(x)$ by a bilinear search,
	\begin{equation}
		\label{bcalgo:step1}
		\begin{aligned}
			& c^* = & \underset{ c\in\mathbb{R}^+, \, L_c(x)\in \mathcal{P}^\text{SOS}}{\text{max}}
			&\:\:c\\
			&  \text{s.t.}
			&  -\frac{\partial V(x)}{\partial x}\dot{x} - L_c(x)(c-V(x)) & \in \mathcal{P}^{\text{SOS}},
		\end{aligned}
	\end{equation}
	\noindent where $L_c(x)$ is an auxiliary factor to relax the non-negativity constraint for the initial barrier function $h(x)$.
	
	\textbf{2.} Search another two auxiliary factors $L_{1,2}(x)$ for $h(x)$,
	\begin{equation}
		\label{bcalgo:step2}
		\begin{aligned}
			&& \exists \; L_1(x), L_2(x)\in\mathcal{P}^{\text{SOS}}&\\
			&  \text{s.t.}
			&  -\frac{\partial V(x)}{\partial x}\dot{x} - L_1(x) h(x)&\in \mathcal{P}^{\text{SOS}}, \\
			&
			&  \frac{\partial h(x)}{\partial x}\dot{x} - L_2(x) h(x) &\in \mathcal{P}^{\text{SOS}}.
		\end{aligned}
	\end{equation}
	
	\noindent Meanwhile, $h(x)$ can be re-written into the square matrix representation form as $h(x)=Z(x)^{\mathrm{T}}QZ(x)$, where $Q$ is a semi-definite coefficient matrix and $Z(x)$ is a monomial vector \cite{wang2018permissive}. The trace $\textit{Tr}(\cdot)$ of $Q$ is used to approximate the volume of barrier certified ROA (BCROA).
	
	\textbf{3.} Enlarge $\textit{Tr}(Q)$ to parameterize a permissive $h(x)$ with fixed $L_1(x)$ and $L_2(x)$,
	\begin{equation}
		\label{bcalgo:step3}
		\begin{aligned}
			&& \underset{ \substack{h(x)\in\mathcal{P}}}{\text{max}} \quad
			\textit{Tr}(Q) &\\
			&  \text{s.t.}
			&  -\frac{\partial V(x)}{\partial x}\dot{x} - L_1(x) h(x)&\in \mathcal{P}^{\text{SOS}}, \\
			&
			&  \frac{\partial h(x)}{\partial x}\dot{x} - L_2(x) h(x) &\in \mathcal{P}^{\text{SOS}}. 
		\end{aligned}
	\end{equation}
	
	\noindent The optimal barrier function $h^*(x)$ can be found with a permissive BCROA if the increase of $\textit{Tr}(Q)$ is less than a threshold, otherwise repeat \textbf{2} and \textbf{3} for a long term. 
	
	These \textbf{3} SOSPs demonstrate the process to obtain an optimal BCROA directly. Thus, if a ROA exists in the learned polynomial system (\ref{eqn:Appro_sys}), a permissive BCROA would be computed by these SOSPs directly.
	
	\subsection{Existence of Probabilistic BCROA}
	The relationship of a BCROA in the learned system (\ref{eqn:gp_learned_system_2}) toward the true dynamics is given in the following theorem.
	
	\begin{theorem}
		\label{them:ROAGP}
		Given $k$ measurements of a partially unknown system (\ref{eqn:Appro_sys}) such that we can obtain a learned system (\ref{eqn:gp_learned_system_2}) with a probability greater or equal to $(1-\delta)^k, \delta\in(0,1)$. Based on the SOSPs (\ref{bcalgo:step1}), (\ref{bcalgo:step2}) and (\ref{bcalgo:step3}), if there exists a BCROA $\mathcal{L}=\{x\in\mathcal{X}\vert h(x)\geq 0\}$ such that
		\begin{equation}
			\label{theo:2}
			\begin{aligned}
				\frac{\partial h(x)}{x}(f(x)+P_k(x)+m_{d_\xi}(x))\geq 0
			\end{aligned}
		\end{equation}
		\noindent holds for the states $x\in\mathcal{L}$, then $\mathcal{L}$ can be regarded as a BCROA in (\ref{eqn:Appro_sys}) with probability bounds $((1-\delta)^k,1)$.
	\end{theorem}

	\begin{proof}		
		Given an arbitrary initial state $x_0\in\mathcal{L}$ in the learned system (\ref{eqn:gp_learned_system_2}), its state trajectory $\psi(t;x_0), t\in[0,\infty]$ is guaranteed to converge to the origin due to the Lyapunov constraints in (\ref{bcalgo:step1}), (\ref{bcalgo:step2}) and (\ref{bcalgo:step3}). Since the derivative of the barrier function $\frac{\partial h(\psi(t; x_0))}{\partial t}$ is non-negative over the trajectory $\psi(t; x_0)$, $h(\psi(t; x_0))$ will strictly increase along $\psi(t; x_0), t\in[0,\infty)$. Then, the BCROA $\mathcal{L}$ established a region with the safety and stability guarantee in (\ref{eqn:gp_learned_system_2}). 
		
		Because $\mathcal{L}$ denotes a certain region in (\ref{eqn:gp_learned_system_2}), while (\ref{eqn:gp_learned_system_2}) is an approximation result of (\ref{eqn:Appro_sys}) within probability bounds $((1-\delta)^k, 1)$. Thus, $\mathcal{L}$ also denotes a potential BCROA toward the true dynamics in (\ref{eqn:Appro_sys}) with the same probability bounds, which ends the proof.
	\end{proof}	
	
	
	The BCROA in Theorem \ref{them:ROAGP} that satisfies (\ref{bcalgo:step1}), (\ref{bcalgo:step2}) and (\ref{bcalgo:step3}) encodes the safety of the learned polynomial system (\ref{eqn:gp_learned_system_2}). It further shows a probabilistic relationship toward the real ROA in the concerned system (\ref{eqn:Appro_sys}) such that we can make high-probability statements about system safety.
	
	\subsection{Probabilistic Optimal BCROA Estimate}
	
	
	To illustrate our procedures of finding probabilistic optimal BCROAs, we present an algorithm below and a corresponding flowchart in Figure 5, which both consist of a learned polynomial system and SOSPs. Besides, the following theorem proposes the confidence range of the result from Algorithm 1 over episodes.
	

	
	
	
	\begin{algorithm}[ht]
		\caption{Optimal BCROAs Estimate}
		\KwIn{Original system (\ref{eqn:sysstate}); initial sublevel set of a given Lyapunov function $\mathcal{L}_0=\{x\vert V(x)\leq c_0\}$; SOSP termination $\epsilon$ and episodes $K$.}
		\KwOut{Optimal BCROAs $R=[\mathcal{L}^*_{1}, \mathcal{L}^*_{2}, \dots, \mathcal{L}^*_{K}]^{\mathrm{T}}$; probability bounds $P=[\delta_1,\delta_2,\dots, \delta_K]^{\mathrm{T}}$.}
    	Preprocess (\ref{eqn:sysstate}) into (\ref{eqn:Appro_sys}) by Chebyshev interpolants. Initialize the prior information set $D_\xi=\emptyset$.\\	
		Store the measurements of $d_\xi(x)$ in (\ref{eqn:Appro_sys}) as $D_0$.\\
		\For{$i\in\{1,2,\dots ,K\}$}{	
			Update prior information $D_\xi = \cup D_{i-1}$.\\	
			Construct (\ref{eqn:gp_learned_system_2}) with polynomial $m(x)$. \\
			Execute SOSP (\ref{bcalgo:step1}) and set the initial barrier function as $h_i(x)=c^*-V(x)$.\\
			Execute SOSPs (\ref{bcalgo:step2}) and (\ref{bcalgo:step3}) to obtain the optimal BCROA $\mathcal{L}^*_i=\{x \vert h^*_i(x) \geq 0\}$.\\ 
			Record the probability bounds $\delta_i$ and update the measurement of $d_\xi(x)$ according to (\ref{eqn:samplepolicy}).
		}	
		\Return $R, P$.\\
	\end{algorithm}
	
	
	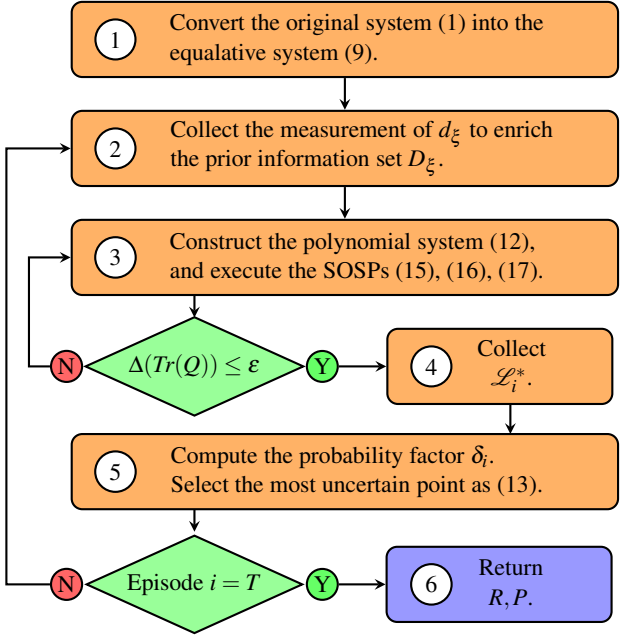
\begin{figure}[ht]
		\centering 
		\begin{tikzpicture}[
			Number/.style ={
				circle,
				draw,
				color=black,
				thick,
				fill=white,
				text=black,
				align=right,
				inner sep=1pt,},
			Number2/.style ={
				circle,
				draw,
				color=black,
				thick,
				fill=white,
				text=black,
				align=center,
				inner sep=2.7pt,},
			DecY/.style ={
				circle,
				draw,
				color=black,
				thick,
				fill=green!60,
				text=black,
				align=center,
				inner sep=1pt,},
			DecN/.style ={
				circle,
				draw,
				color=black,
				thick,
				fill=red!60,
				text=black,
				align=center,
				inner sep=1pt,}]
			
			\node (start1) [start, align=justify] { 
                \quad Convert the original system (\ref{eqn:sysstate}) into the\\ 
                \quad equalative system (\ref{eqn:Appro_sys}).};
			\node (start2) [start, align=justify, below of=start1, yshift=-0.45cm] { 
				\quad\, Collect the measurement of $d_\xi$ to enrich\\
				\quad\, the prior information set $D_\xi$.};
			\node (start3) [start, align=justify, below of=start2, yshift=-0.45cm] 
			{\quad Construct the polynomial system (\ref{eqn:gp_learned_system_2}), \\ 
				\quad and execute the SOSPs (\ref{bcalgo:step1}), (\ref{bcalgo:step2}), (\ref{bcalgo:step3}).};			
			\node (dec1)  [decision, below of=start3, xshift=-2.cm, yshift=-0.45cm] {$\Delta(\textit{Tr}(Q))\leq \epsilon$};
			\node (proc1) [proc, align=center, below of=start3, xshift=2.06cm, yshift=-0.45cm] 
			{\quad Collect\\ 
				\quad $\mathcal{L}^*_{i}$.};
			\node (start4) [start, align=justify, below of=start3, yshift=-1.85cm] 
			{\quad Compute the probability factor $\delta_i$.\\ 
				\quad Select the most uncertain point as (\ref{eqn:samplepolicy}).};
			\node (dec2)  [decision, below of=start4, xshift=-2cm, yshift=-0.5cm] {Episode $i=T$};
			\node (end) [end, align=center, below of=start4, xshift=2.05cm, yshift=-.5cm] {\quad Return\\ \quad $R, P$.};
			\node[Number2](number1) at (-3.05,0) {1};
			\node[Number2](number2) at (-3.05,-1.45) {2};
			\node[Number2](number3.a) at (-3.05,-2.9) {3};
			\node[DecN](decN1) at (-3.7,-4.35) {N};
			\node[DecY](decY1) at (-0.3,-4.35) {Y};
			\node[Number2](number3.b) at (1.15,-4.35) {4};
			\node[DecY](decY2) at (-0.3,-7.25) {Y};
			\node[Number2](number3.c) at (-3.05,-5.75) {5};
			\node[DecN](decN2) at (-3.7,-7.25) {N};
			\node[Number2](number3.c) at (1.15,-7.25) {6};
			\draw [arrow] (start1) -- (start2);
			\draw [arrow] (start2) -- (start3);
			\draw [arrow] (-2,-3.4) -- (-2,-3.7);
			\draw [arrow] (decY1) -- (proc1);
			\draw [arrow] (decN1) -- (-4.2,-4.35) -- (-4.2,-2.9) -- (start3);
			\draw [arrow] (2.2,-4.85) -- (2.2,-5.25);
			\draw [arrow] (-2,-6.25) -- (-2,-6.55);
			\draw [arrow] (decY2) -- (end);
			\draw [arrow] (decN2) -- (-4.5,-7.25) -- (-4.5,-1.45) -- (start2);
		\end{tikzpicture}
		\label{fig:flowchart}
		\caption{Flowchart of Algorithm 1.}
	\end{figure}
	
	\begin{theorem}
		\label{them:probROA}
		Algorithm 1 establishes $K$ probabilistic ROA $[\mathcal{L}^*_1,\mathcal{L}^*_2,\dots, \mathcal{L}^*_K]^{\mathrm{T}}$ with probabilities greater or equal to $[(1-\delta_1)^{n_1},(1-\delta_2)^{n_2},\dots,(1-\delta_K)^{n_K}]^{\mathrm{T}}$, where $n_{i}$ denotes the sample number with index $i=1,2\cdots,K$. Then, the following probability bounds are given:
		\begin{equation}\label{them:probROAS}\small
			\begin{aligned}
				\mathbb{P}\left\{\forall i \in [1,K], \forall \mathcal{L}^{*}_{i} \neq \emptyset, \mathcal{L}^*_{\cup} \subset D \right\}&=\min\{ (1-\delta_{i})^{n_i}\}\\
			    \mathbb{P}\left\{\forall i \in [1,K], \forall \mathcal{L}^{*}_{i} \neq \emptyset, \mathcal{L}^*_{\cap}\subset D \right\}&=\max\{(1-\delta_{i})^{n_i}\},
			\end{aligned}
		\end{equation}
		\noindent where $\mathcal{L}^*_{\cup}= \cup_{i=1}^K \mathcal{L}^*_i$ and $\mathcal{L}^*_{\cap}=\cap _{i=1}^K\mathcal{L}^*_i$ denote the union set and the intersection set of these optimal BCROAs, respectively, and $D$ denotes the exact ROA of the system (\ref{eqn:Appro_sys}).	
	\end{theorem}
	\begin{proof}
		After executing the algorithm $K$ times, we could get optimal BCROAs $[\mathcal{L}^*_1,\mathcal{L}^*_2,\dots, \mathcal{L}^*_K]^{\mathrm{T}}$ with probability bounds $[[(1-\delta_1)^{n_1},1), [(1-\delta_2)^{n_2},1), \dots, [(1-\delta_K)^{n_K},1)]^{\mathrm{T}}$ toward the dynamical system (\ref{eqn:Appro_sys}). Since the increasing $K$ generates larger information set $D_\xi$, when $K\rightarrow \infty$, the GP model can be inferred with the richest $D_\xi$ such that the lower bound of this probabilistic range is approximating $1$. Let $\mathcal{L}^*_{\cap}$ be the intersection set of $\mathcal{L}^*_i$, which represents the most common part of these optimal BCROAs toward the exact dynamics of (\ref{eqn:Appro_sys}). Then, a conservative statement about $\mathcal{L}^*_{\cap}$ toward the real ROA $D$ in (\ref{eqn:Appro_sys}) establishes with a minimum value of $\{(1- \delta_{i})^{n_K}\}$. Therefore, the second equation in (\ref{them:probROAS}) about $\mathcal{L}^*_{\cup}$ toward $D$ establishes in a similar way, which completes the proof. 
	\end{proof}
	
	Theorem \ref{them:probROA} presents a probabilistic statement of $\mathcal{L}^*_i$ generated from Algorithm 1. Obviously, a higher degree mean function $m_{d_\xi}(x)$ in (\ref{eqn:gp_learned_system_2}) could approximate the real dynamics much more precisely. Note that the shape of $\mathcal{L}^*_i$ is various with uncertainty. It is worthy noting that the corresponding Lyapunov function $V(x)$ and the barrier function $h(x)$ could increase their degrees simultaneously for a better estimate. This is because a higher degree polynomial certified ROA can explore more areas, so the probabilistic bound narrows naturally.
	
	\section{Numerical Examples}
	Based on the Matlab toolbox of Chebfun, GPML, SOSOPT and Mosek solver, we estimate the BCROA of the autonomous partially unknown system by two examples.
	
	\subsection{Example 4: A 2D Nonlinear System}
	Extended to Example 3, we consider the system:
	\begin{eqnarray}\label{demo:2D}
		\begin{aligned}
			\begin{bmatrix} \dot{x}_1 \\ \dot{x}_2 \end{bmatrix} = 
			\begin{bmatrix}
				-x_{1}+x_{2}\\	x_{1}^2x_{2}+1-\sqrt{|\exp{(x_{1})}\cos{(x_{1})}|}
			\end{bmatrix}+
			\begin{bmatrix}
				0\\	d(x)
			\end{bmatrix}.
		\end{aligned}   
	\end{eqnarray}
	
	\noindent Example 3 shows the construction of the first learned dynamical system in Figure \ref{fig:fig_Comparison}(b) and the noise over the system measurement is bounded by $0.01$. The learned GP models are consisted by a $2^{nd}$ degree polynomial mean function zero, and a RBF kernel function $k(x,x^{\prime})$ with signal variances $\exp(0.1)$ and length scale $\exp(0.2)$. 
	
	At each episode, a given LCROA $\{x \vert \;V(x)\leq c_0 = 0.1\}$ is considered to the next step, where $V(x)=2x_{1}^4+0.5x_{2}^4-x_{1}^2x_{2}^2+x_{1}^2+x_{2}^2-0.5x_{1}x_{2}$ is a $4^{th}$ degree polynomial. Besides, we obtained the probability bounds with a fixed discounting factor $\sqrt{\beta}=2$ according to the method in \cite{jagtap2020control}. After $10$ episodes, the probability bounds of these optimal BCROAs is $\mathbb{P}\{\forall \mathcal{L}_{i\in[1,10]}^{*}, \mathcal{L}^*_i \subset D\}\in [0.9500,0.9700]$, where $D$ denotes the exact ROA in (\ref{demo:2D}).

	\begin{figure}[ht] 
		\centering
		\includegraphics[width=0.9\linewidth]{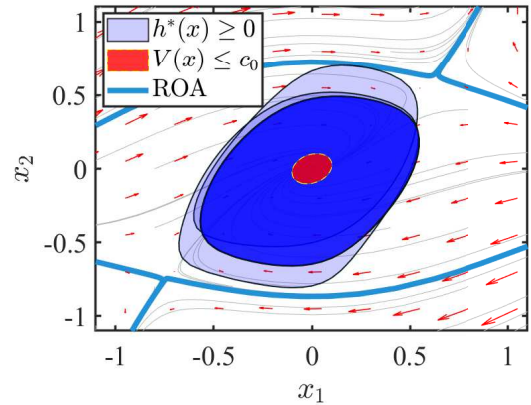}
		\caption{ROA estimate of (\ref{demo:2D}). The red region enclosed by the yellow dashed line depicts the LCROA and the blue region enclosed by the black solid line depicts the optimal BCROA. The gray lines and the red arrows depict the corresponding vector field of (\ref{demo:2D}) without unknown term $d(x)$. }
		\label{fig:fig2d_1}
	\end{figure}
	
	In Figure \ref{fig:fig2d_1}, each computed optimal BCROA is displayed with a fixed transparency. The intersection region of these optimal BCROAs is established with the probability $[0.9700,1.0000]$, while the union region of these optimal BCROAs is established with the probability $[0.9500,1.0000]$. Therefore, the result of these probabilistic optimal BCROAs in Figure \ref{fig:fig2d_1} is not only obviously larger than the LCROA, but also guaranteeing the system safety by a clear probability range.
	
	\subsection{Example 5: A 3D Nonlinear System}
	Consider a 3D nonlinear system corrupted with unknown terms $d_1(x), d_3(x)$ as the form of (\ref{eqn:sysstate}),
	\begin{small}
		\begin{eqnarray}\label{demo:3D}
			\begin{aligned}
				\begin{bmatrix} \dot{x}_1 \\ \dot{x}_2\\\dot{x}_3 \end{bmatrix} = 
				\begin{bmatrix}
					-x_{1}^2-\cos{(x_{1}^2)}\sin{(x_{1})}\\	
					-x_{2}-x_{1}^3x_{2}\\
					-x_{1}^2x_{3}+1-\sqrt{|\exp{(x_{1})}\cos{(x_{1})}|}
				\end{bmatrix}+
				\begin{bmatrix}
					d_1(x)\\0\\	d_3(x)
				\end{bmatrix}.
			\end{aligned}   
		\end{eqnarray}
	\end{small}
	
	\noindent The non-polynomial terms in (\ref{demo:3D}) are approximated by the $4^{th}$ degree Chebyshev interpolants in $[-2,2]^3$. To learn the system dynamics, we first collect the trajectory starting at $(-0.1, -0.1,0.1)$, where the noises over $d_1(x)$ and $d_3(x)$ are bounded by $0.01$. In this case, we use a $3^{rd}$ degree polynomial mean function and a RBF kernel with signal variances $\exp(0.01)$ and length scale $\exp(0.2)$.
	
	A sublevel set $\{x \vert \;V(x)\leq c_0 = 5\times 10^{-3}\}$ is given to construct the LCROA with a $4^{th}$ order $V(x)=10x_1^4+x_2^4+20x_3^4+2x_1^2x_2^2-4x_3^2x_2^2+3x_1^2x_3^2$. Set $\sqrt{\beta}=2$ such that the probability bounds of $10$ optimal BCROAs can be computed as $\mathbb{P}\{\forall \mathcal{L}_{i=1}^{10}, \; \mathcal{L}_i \subset D\}\in [0.8395,0.9567]$. 
	
	These optimal BCROAs are displayed in Figure \ref{fig:fig3D} with a fixed transparency. From Figure \ref{fig:fig3D}, we see that there are various differences of these optimal BCROAs. Furthermore, as seen in Table 1, with the same hyperparameters, the computation time of Example 5 is greater than Example 4 significantly. However, at the episode 10, Example 5 will first stop learning due to the high computation cost. Although higher dimensional functions find a comparable BCROA, but this is often not the case due to the fragile and changeable dynamics under uncertainty.
	
	\begin{figure}[ht] 
		\centering
		\includegraphics[width=0.95\linewidth]{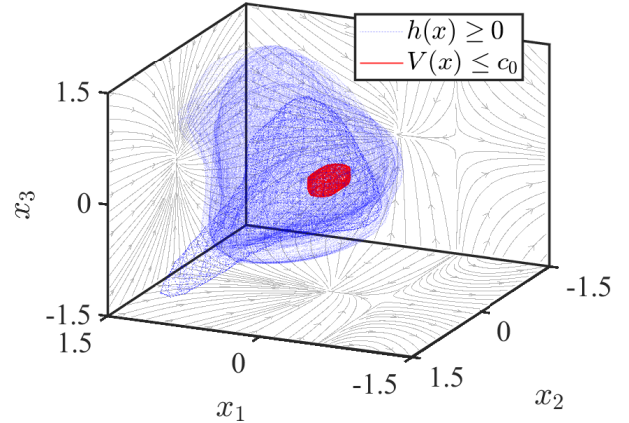}
		\caption{ROA estimates of (\ref{demo:3D}). The red cube enclosed by the solid lines depicts the LCROA and the blue cube enclosed by the dashed line denotes optimal BCROAs. The vector field is projected onto the plane.}
		\label{fig:fig3D}
	\end{figure}
	
	\begin{table}
		\centering
		\caption{Computation Time $t_c$ {\rm [sec]} of Algorithm 1.} \label{tab:test}
		\begin{tabular}{cccc}
			\hhline 
			&Episode 1 &Episode 5 &Episode 10\\ \hline
			Example 4 \; \vline& $35.6 $ & $168.3 $ & $1752.6 $ \\ 
			Example 5 \; \vline& $406.4$ & $837.8 $ & $1329.4$ \\ 
			\hhline
		\end{tabular}
	\end{table}
	
	\section{Conclusion}\label{sec:conclusions}
	In this work, a method is proposed to compute a barrier certified region of attraction such that the stability and safety can both be guaranteed for a class of partially unknown systems. The proposed method is built based on Chebyshev interpolants, Gaussian processes, sum-of-squares programming and a safe sample policy. The effectiveness of the proposed algorithm has been demonstrated via two numerical examples.
	
On the other hand, we would like to admit that a large amount of unknown terms may affect the accuracy and efficiency of proposed algorithm. Possible solutions to address this issue could be to construct an efficient sample policy \cite{umlauft2020smart} or to exploit local GP regression \cite{nguyen2009model}. Thus, future efforts will be devoted to prior data selection in further enlarging the estimated region of attraction for partially unknown systems \cite{wang2018permissive, umlauft2018uncertainty}.
	

\section*{Appendix: Proof of Proposition 1}
\begin{proof}
	First, Mercer's Theorem \cite{konig2013eigenvalue} allows us to represent a kernel by a kernel basis vector $\phi(\cdot)$. For this paper, we restrict our attention to the RBF kernel $k(x,x^{\prime})$, though our methodology can support more general kernels,
	\begin{small}
		\begin{equation}
			\label{eqn:GP_kernel}
			\begin{aligned}
				k(x,x^{\prime})&=\sigma^2\exp(\frac{(x-x^{\prime})^2}{-2l^2})\\
				&=\sum\limits_{i=0}^{\infty}\frac{\sigma^22^i}{2l^2i!}\exp(\frac{(x)^2}{-2l^2})(\frac{x}{2l^2})^i\exp(\frac{(x')^2}{-2l^2})(\frac{x'}{2l^2})^i\\
				&=\phi(x)^{\mathrm{T}}\phi(x'),
			\end{aligned}
		\end{equation}
	\end{small}
	
	\noindent where $\sigma$ denotes the signal variance, $l$ denotes the length scale and $\phi(x)$ denotes an infinite dimensional vector,
		\begin{equation}\small
			\label{eqn:GP_kernel_Basisfunction_2}
			\begin{aligned}
				\phi(x)=\exp(\frac{(x)^2}{-2l^2})\cdot(\sqrt{\frac{\sigma^22^0}{2l^20!}}(\frac{x}{2l^2})^0,\sqrt{\frac{\sigma^22^1}{2l^21!}}(\frac{x}{2l^2})^1,\dots).
			\end{aligned}
		\end{equation}
	Then, since $d(x)$ and its measurement $y(x)$ in (\ref{eqn:sysstate}) are bounded in the RKHS, we can formulate,
	\begin{equation}\small
		\label{eqn:GP_prior}
		\begin{aligned}
			d(x)= \phi(x)^{\mathrm{T}}w, \quad 	y(x) &= d(x)+\epsilon,
		\end{aligned}
	\end{equation}
	\noindent where $w$ denotes the weight vector and $\epsilon$ denotes the noise. Let $p(w)\sim\mathcal{N}(0,\sigma_p^2)$ be the prior distribution of $w$ and let $p(w\vert y, x)$ be the posterior distribution of $w$ based on $y$ and $x$, we can infer their relationship as follows,
	
	\begin{equation}\small
	\begin{small}\label{straincomponent}
		\text{posterior} = \frac{\text{likelihood}\times \text{prior}}{\text{marginal\;likelihood}}, p{(w \vert y, x)}= \frac{p(y \vert x, w)\cdot p(w)}{p(y\vert x)}.
	\end{small}
	\end{equation}
	
	\noindent Next, let $K_1$ be the weight-independent marginal likelihood $p(y\vert x)$. GP is relocating likelihood $p(y \vert x, w)$ with measurements $[(x_1,y_1), (x_2,y_2), \dots, (x_k, y_k)]^{\mathrm{T}}$ as below,
	\begin{equation}\small
		\label{eqn:likelihood}
		\begin{aligned}
			p(y \vert x, w) &= \prod_{i=1}^{k} p \left( y_{i}, x_i, w\right) = \prod^{k}_{i=1}\frac{\exp{(\frac{(y_i - \phi(x_i)^{\mathrm{T}} w)^2}{-2\sigma_n^2})}}{\sigma_n\sqrt{2\pi}}, \\
		\end{aligned}
	\end{equation}
	\noindent to compute the $p{(w \vert y, x)}$ based on (\ref{straincomponent}) and (\ref{eqn:likelihood}),
	\begin{equation}\small
		\label{eqn:posterior_weights}
		\begin{aligned}
			p(w \vert x, y) 
			& = \frac{\prod\limits_{i=1}^{k} p \left( y_{i}, x_i, w\right) \cdot \frac{1}{\sqrt{2\pi \sigma^2_p}}\exp{(\frac{w^{\mathrm{T}}\sigma_p^{-2}w}{-2})}}{K_{1}}\\
			& = \!K_2 \cdot \exp{(\frac{w^{\mathrm{T}}\sigma_p^{-2}w}{-2})}\cdot \exp{(\frac{\sum\limits_{i=1}^k(y_i -\phi(x_i)^{\mathrm{T}}w)^2}{-2\sigma_n^2})}, \\
		\end{aligned}
	\end{equation}
	\noindent where $K_2=\frac{1}{K_1\cdot \sqrt{2\pi\sigma_p^2}\cdot (\sigma_n\sqrt{2\pi})^{k}}$. 
	
	By using the Maximum a Posterior estimation, the exact value of the optimal weight $\hat{w}$ can be computed as,
	\begin{small}
		\begin{equation}
			\label{eqn:optimize}
			\begin{aligned}
				\hat{w} &= \arg \max_{w}\; p(w\vert x, y) = \arg \max_{w} \log{ p(w\vert x, y)}\\
				& = \arg \min_{w}\; -\log{ p(w\vert x, y)}\\
				& = \arg \min_{w}\; -\log{K_2} +\frac{w^{\mathrm{T}}\sigma_p^{-2}w}{2}+\frac{\sum\limits_{i=1}^k(y_i -\phi(x_i)^{\mathrm{T}}w)^2}{2\sigma_n^2}\\
				& = \arg \min_{w}\; w^{\mathrm{T}}(\sigma_n^2\sigma_p^{-2})w+\sum\limits_{i=1}^{k}(y_i -\phi(x_i)^{\mathrm{T}}w)^2\\
				& = \arg \min_{w}\; w^{\mathrm{T}}(\phi(x)^{\mathrm{T}}\phi(x)+\sigma_n^2\sigma_p^{-2})w-2w^{\mathrm{T}}\phi(x)^{\mathrm{T}}y+y^{\mathrm{T}}y\\
				& = \arg \min_{w}\; L(w).
			\end{aligned}
		\end{equation}	
	\end{small}
	
	The derivative of $L(w)$ in (\ref{eqn:optimize}) can be computed as,
	\begin{equation}\small
		\label{eqn:optimize_derivative}
		\begin{aligned}
			\frac{\partial L(w)}{\partial w}&=2(\phi(x)^{\mathrm{T}}\phi(x)+\sigma_n^2\sigma_p^{-2})w-2\phi(x)^{\mathrm{T}}y,
		\end{aligned}
	\end{equation}	
	\noindent and the optimal weight $\hat{w}$ can be parameterized directly,
	\begin{equation}
		\label{eqn:optimize_w}\small
		\begin{aligned}
			\hat{w}&=(\phi(x)^{\mathrm{T}}\phi(x)+\sigma_n^2\sigma_p^{-2})^{-1}\phi(x)^{\mathrm{T}}y\\
			&=\phi(x)^{\mathrm{T}}(\phi(x)\phi(x)^{\mathrm{T}}+\sigma_n^2\sigma_p^{-2})^{-1}y\\
			&=\sigma_p^2\phi(x)^{\mathrm{T}}(\phi(x)\sigma_p^2\phi(x)^{\mathrm{T}}+\sigma_n^2)^{-1}y,
		\end{aligned}
	\end{equation}	

	\noindent
    Thus, we can obtain the optimal $p(\hat{w}\vert x,y)$ as (\ref{eqn:posterior_weights}) shows,
	\begin{equation}\small
		\label{eqn:optimize_w2}
		\begin{aligned}
			p(\hat{w}\vert x,y)\sim(&\sigma_p^2\phi(x)^{\mathrm{T}}(\phi(x)\sigma_p^2\phi(x)^{\mathrm{T}}+\sigma_n^2)^{-1}y,\\
			&(\phi(x)\sigma_p^2\phi(x)^{\mathrm{T}}+\sigma_n^2)^{-1}),
		\end{aligned}
	\end{equation}	
	
    \noindent where the predictive output $d(x_*)$ of the query states $x_*$ is,
	\begin{equation}\small
		\label{eqn:predictive_output}
		\begin{aligned}
 			d(x_*) &=\phi(x_*)^{\mathrm{T}}\hat{w}.
		\end{aligned}
	\end{equation}	
	
    Because $\phi(\cdot)$ is an infinite vector, which is used to approximate $y$ with optimal weight $\hat{w}$. To avoid the loss of generality, we can truncate $\phi(x)$ by a finite monomial vector $\varphi(x)$ without any impact of the inference processes from (\ref{eqn:likelihood}) to (\ref{eqn:optimize_w2}). This will reshape the mean function in the polynomial form and adjust the second term in (\ref{eqn:optimize_w2}) as,
	
	\begin{equation}\small
		\label{eqn:optimize_output_distribution_mean2}
		\begin{aligned}
			m(x_*) &= \varphi(x_*)^{\mathrm{T}} \hat{w}, \\
			\sigma^2(x_*) &= k(x_*, x_*)-k_*^{\mathrm{T}}(K+\sigma_n^2I)^{-1}k_*.
		\end{aligned}
	\end{equation}	
	\noindent where $[K]_{(i,j)}=k(x_i, x_j)$ is a kernel Gramian matrix and $k_*=[k(x_1,x_*), k(x_2,x_*), \dots, k(x_w, x_*)]^{\mathrm{T}}$. Thus, Proposition \ref{prop:polynomial mean function} is obtained, which completes the proof.
\end{proof}
	
\end{document}